\documentclass[a4paper]{article}

\usepackage[utf8]{inputenc}
\usepackage{amssymb,amsmath,amsthm}
\usepackage{xifthen}
\usepackage{xparse}
\usepackage{arydshln}
\usepackage{ifmtarg}
\usepackage{graphicx}

\title{On the \DIS{} meta-problem and applications to the complexity of identifying problems on graphs.}
\author{Barbero Florian, Isenmann Lucas, Thiebaut Jocelyn \\ LIRMM, Université de Montpellier, \\161 rue Ada, 34095, Montpellier, France \\{\tt \{florian.barbero, lucas.isenmann, jocelyn.thiebaut\}@lirmm.fr}}

\theoremstyle{plain}
\newtheorem{theorem}{Theorem}
\newtheorem{lemma}{Lemma}
\newtheorem{corollary}[lemma]{Corollary}
\newtheorem{proposition}[theorem]{Proposition}

\theoremstyle{definition}
\newtheorem{definition}{Definition}

\newenvironment{qproblem}[3]{
\textsc{#1}\\
\textbf{Input:} #2\\
\textbf{Output:} #3
}

\newcommand{\NP}{{\sf NP}}
\newcommand{\FPT}{{\sf FPT}}
\newcommand{\W}[1]{{\sf W[#1]}}

\newcommand{\aDIS}{{\sc DIS}}
\newcommand{\DIS}{{\sc Distance Identifying Set}}

\newcommand{\aIC}{{\sc IC}}
\newcommand{\IC}{{\sc Identifying Code}}
\newcommand{\aLD}{{\sc LD}}
\newcommand{\LD}{{\sc Locating Dominating Set}}
\newcommand{\aMD}{{\sc MD}}
\newcommand{\MD}{{\sc Metric Dimension}}

\newcommand{\HS}{{\sc Hitting Set}}

\newcommand{\layered}{layered}

\DeclareMathOperator{\true}{true}
\DeclareMathOperator{\false}{false}
\newcommand{\intset}[1]{[\![#1]\!]}
\newcommand{\inftyset}{\intset{\infty}}
\newcommand{\len}[1]{b_{#1}}

\newcommand{\HIC}{H_r}
\newcommand{\BIC}{B_r}
\newcommand{\CIC}{C_r}

\newcommand{\Set}{\mathcal{S}}
\newcommand{\ov}{\bar{v}}

\begin{document}

\maketitle
\begin{abstract}
 Numerous problems consisting in identifying vertices in graphs using distances are useful in domains such as network verification and graph isomorphism. Unifying them into a meta-problem may be of main interest. We introduce here a promising solution named {\sc Distance Identifying Set}. The model contains {\sc Identifying Code} ({\sc IC}), {\sc Locating Dominating Set} ({\sc LD}) and their generalizations $r$-{\sc IC}{} and $r$-{\sc LD} where the closed neighborhood is considered up to distance $r$. It also contains {\sc Metric Dimension} ({\sc MD}) and its refinement $r$-{\sc MD} in which the distance between two vertices is considered as infinite if the real distance exceeds $r$. Note that while \mbox{{\sc IC}{} = $1$-{\sc IC}{}} and {\sc LD} = $1$-{\sc LD}, we have {\sc MD} = $\infty$-{\sc MD}; we say that {\sc MD} is not \textit{local}.\medskip

 In this article, we prove computational lower bounds for several problems included in \linebreak {\sc Distance Identifying Set} by providing generic reductions from \textsc{(Planar) \HS{}} to the meta-problem. We mainly focus on two families of problem from the meta-problem: the first one, called \textit{bipartite gifted local}, contains $r$-{\sc IC}{}, $r$-{\sc LD} and $r$-{\sc MD} for each positive integer $r$ while the second one, called \textit{$1$-\layered{}}, contains {\sc LD}, {\sc MD} and $r$-{\sc MD} for each positive integer $r$. We have:
 
 \begin{itemize}
  \item the $1$-\layered{} problems are \NP-hard even in bipartite apex graphs,
  \item the bipartite gifted local problems are \NP-hard even in bipartite planar graphs,
  \item assuming ETH, all these problems cannot be solved in $2^{o(\sqrt{n})}$ when restricted to bipartite planar or apex graph, respectively, and they cannot be solved in $2^{o(n)}$ on bipartite graphs,
  \item even restricted to bipartite graphs, they do not admit parameterized algorithms in $2^{\mathcal{O}(k)} \cdot n^{\mathcal{O}(1)}$ except if $\W{0} = \W{2}$. Here $k$ is the solution size of a relevant identifying set.
 \end{itemize}
 In particular, {\sc Metric Dimension} cannot be solved in $2^{o(n)}$ under ETH, answering a question of Hartung in~\cite{HartungN13}.
\end{abstract}

\newpage

\section{Introduction and Corresponding Works}\label{sec:intro}
Problems consisting in identifying each element of a combinatorial structure with a hopefully small number of elements have been widely investigated.
Here, we study a meta identification problem which generalizes three of the most well-known identification problems in graphs, namely \IC{} (\aIC{}), \LD{} (\aLD{}) and \MD{} (\aMD{}).
These problems are used in network verification \cite{BampasBDGKP15,BeerliovaEEHHMR06}, fault-detection in networks \cite{KarpovskyCL98,UngrangsiTS04}, graph isomorphism \cite{Babai80} or logical definability of graphs \cite{KimPSV05}.
The versions of these problems in hypergraphs have been studied under different names in  \cite{BollobasS07}, \cite{BONDY1972201} and \cite{CharbitCCHL08}.\medskip

\noindent Given a graph $G$ with vertex set $V$, the classical identifying sets are defined as follows:
\begin{itemize}
 \item \aIC: Introduced by Karposky et al. \cite{KarpovskyCL98}, a set $C$ of vertices of $G$ is said to be an \emph{identifying code} if none of the sets $N[v] \cap C$ are empty, for $v \in V$ and they are all distinct.
       
 \item \aLD: Introduced by Slater \cite{Slater87,slater1988dominating}, a set $C$ of vertices of $G$ is said to be a \emph{locating-dominating set} if none of the sets $N[v] \cap C$ are empty, for $v \in V\setminus C$ and they are all distinct.
       When not considering the dominating property ($N[v] \cap C$ may be empty), these sets have been studied in \cite{Babai80} as distinguing sets and in \cite{KimPSV05} as sieves.
       
 \item \aMD: Introduced independently by Harary et al. \cite{harary1976metric} and Slater \cite{Slater75}, a set $C$ of vertices of $G$ is said to be a \emph{resolving set} if $C$ contains one vertex from each connected component of $G$ and, for every distinct vertices $u$ and $v$ of $G$, there exists a vertex $w$ of $C$ such that $d(w,u) \neq d(w,v)$.
       The \emph{metric dimension} of $G$ is the minimum size of its resolving sets.
\end{itemize}

The corresponding minimization problems of the previous identifying sets are defined as follows: given a graph $G$, compute a suitable set $C$ of minimal size, if one exists. In this paper, we mainly focus on the computational complexity of these minimization problems.

\subparagraph*{Known results.} A wide collection of \NP-hardness results has been proven for the problems.

For \aIC{} and \aLD{}, the minimization problems are indeed \NP-hard \cite{CohenHLZ99,ColbournSS87}.
Charon et al. showed the \NP-hardness when restricted to bipartite graphs \cite{CharonHL03}, while Auger showed it for planar graphs with arbitrarily large girth \cite{Auger10}. For trees, there exists a linear algorithm \cite{Slater87}.

\MD{} is also \NP-hard, even when restricted to Gabriel unit disk graphs \cite{GareyJ79,HoffmannW12}.
Epstein et al. \cite{EpsteinLW15} showed that \aMD{} is polynomial on several classes as trees, cycles, cographs, partial wheels, and graphs of bounded cyclomatic number, but it remains \NP-hard on split graphs, bipartite graphs, co-bipartite and line graphs of bipartite graphs.
Additionally, Diaz et al. \cite{DiazPSL12} proved a quite tight separation: the problem is polynomial on outerplanar graphs whereas it remains \NP-hard on bounded degree planar graphs.

In a recent publication, Foucaud et al. \cite{FoucaudMNPV17} also proved the \NP-hardness of the three problems restricted to interval graphs and permutation graphs.\medskip

These notions may be considered under the parameterized point of view; see \cite{DowneyF99} for a comprehensive study of Fixed Parameter Tractability (\FPT). In the following, the parameter $k$ is chosen as the solution size of a suitable set.

For \aIC{} and \aLD{}, the parameterized problems are clearly \FPT{} since the number of vertices of a positive instance is bounded by $2^k+k$ ($k$ vertices may characterize $2^k$ neighbors).

Such complexity is not likely to be achievable in the case of \aMD{}, since it would imply $\W{2} = \FPT~(= \W{0})$. Indeed, Hartung et al. \cite{Hartung14,HartungN13} showed \aMD{} is $\W{2}$-hard for bipartite subcubic graphs. The problem is however \FPT{} on families of graphs with degree $\Delta$ growing with the number of vertices because the size $k$ of a resolving set must satisfy $log_3(\Delta) < k$.
Finally, Foucaud et al. \cite{FoucaudMNPV17} provided a \FPT{} algorithm on interval graphs.

\subparagraph*{Our contributions.}

\begin{figure}
 \centering
 \makebox[\textwidth][c]{
  \begin{tabular}{c||c|c|c}
                             & $1$-\layered{}                                                                                                          & $r$-local $0$-\layered{}  & $r$-local                 \\
                             & problems                                                                                                                & problems                  & problems                  \\ \hline \rule{0pt}{5ex}
   
   \begin{minipage}{0.1\textwidth} \centering \textsc{Planar} \HS{}\end{minipage} & \begin{minipage}{0.22\textwidth} \centering \NP-hard on\\ bipartite apex graphs \end{minipage}                                                                                               & \begin{minipage}{0.24\textwidth} \centering \NP-hard on\\ bipartite planar graphs \end{minipage} & \begin{minipage}{0.27\textwidth} \centering (bipartite) planar gadget $\Rightarrow$ \NP-hard on (bipartite) planar graphs \end{minipage} \\[3ex] \hdashline
   \rule{0pt}{3.5ex}
   with ETH                  & \multicolumn{3}{c}{no algorithm running in $2^{\mathcal{O}(\sqrt{n})}$ time for relevant classes of graphs.}                                                                    \\ \hline \rule{0pt}{5ex}
   
   \HS{}                     & \multicolumn{2}{c|}{\NP-hard on bipartite graph}                                                                        & \begin{minipage}{0.25\textwidth} \centering (bipartite) gadget $\Rightarrow$\\ \NP-hard on (bipartite) graphs \end{minipage}                             \\[3.5ex] \hdashline
   \rule{0pt}{3.5ex}
   with ETH                  & \multicolumn{3}{c}{no algorithm running in $2^{\mathcal{O}(n)}$ for (bipartite) graphs.}                                                                                        \\
   with $\W{2} \neq \W{0}$   & \multicolumn{3}{c}{no parameterized algorithm in $2^{\mathcal{O}(k)} \cdot n^{\mathcal{O}(1)}$ for (bipartite) graphs.}                                                         \\ \hline
  \end{tabular}
 }
 \caption{The computational lower bounds implied by our generic reductions.}
\end{figure}
In order to unify the previous minimization problems, we introduce the concept of \textit{distance identifying functions}. Given a distance identifying function $f$ and a value $r$ as a positive integer or infinity, the \DIS{} meta-problem consists in finding a minimal sized $r$-dominating set which distinguishes every couple of vertices of an input graph thanks to the function $f$. Here, we mainly focus on two natural subfamilies of problems of \DIS{} named \textit{local}, in which a vertex cannot discern the vertices outside of its $i$-neighborhood, for $i$ a fixed positive integer, and \textit{$1$-\layered{}}, where a vertex is able to separate its open neighborhood from the distant vertices.

\smallskip
With this approach, we obtain several computational lower bounds for problems included in \DIS{} by providing generic reductions from \textsc{(Planar) \HS{}} to the meta-problem. The reductions rely on the set/element-gadget technique, the noteworthy adaptation of the clause/variable-gadget technique from \textsc{SAT} to \HS{}.\medskip

As we provide a $1$-\layered{} generic gadget, the $1$-\layered{} reductions operate without condition. For local problems, the existence of a local gadget is not always guaranteed. Thus, a local reduction operates only if a local gadget is provided. However, the local planar reduction is slightly more efficient than its $1$-\layered{} counterpart: it indeed implies computational lower bounds for planar graphs whereas the $1$-\layered{} reduction requires an auxiliary apex, limiting the consequences to apex graphs.

The reductions in general graphs are designed to exploit the $\W{2}$-hardness of \HS{} parameterized by the solution size $k_{HS}$ of an hitting set, hereby using:\medskip
\begin{theorem}[folklore] \label{thm:w2bound}
 Let $n_{HS}$ and $m_{HS}$ be the number of elements and sets of an \HS{} instance, and $k_{HS}$ be its solution size. A parameterized problem with parameter $k$ admitting a reduction from \HS{} verifying $k = \mathcal{O}(k_{HS} + log(n_{HS} + m_{HS}))$ does not have a parameterized algorithm running in $2^{\mathcal{O}(k)} \cdot n^{\mathcal{O}(1)}$ time except if $\W{2} = \FPT{}$.
\end{theorem}
\begin{proof}
 Given a reduction from \textsc{Hitting Set} to a parameterized problem $\Pi$ such that the reduced parameter satisfies $k = \mathcal{O}(k_{HS} + log(n_{HS} + m_{HS}))$ and the size of the reduced instance verifies $n = (n_{HS} + m_{HS})^{\mathcal{O}(1)}$, an algorithm for $\Pi$ of running time $2^{\mathcal{O}(k)} \cdot n^{\mathcal{O}(1)}$ is actually an algorithm for \textsc{Hitting Set} of running time $2^{\mathcal{O}(k_{HS})} \cdot (n_{HS} + m_{HS})^{\mathcal{O}(1)}$, meaning that \HS{} is \FPT, a contradiction to its $\W{2}$-hardness (otherwise $\W{2} = \FPT$).
\end{proof}
Hence, as each gadget contributes to the resulting solution size of a distance identifying set, we set up a binary compression of the gadgets to limit their number to the logarithm order. From the best of our knowledge, this merging gadgets technique has never been employed.\medskip

\medskip
The organization of the paper is as follows. After a short reminder of the computational properties of \HS{}, Section~\ref{sec:def} contains the definitions of distance identifying functions and sets, allowing us to precise the computation lower bounds we obtain. The Section~\ref{sec:sup} designs the supports of the reductions as \textit{distance identifying graphs} and \textit{compressed graph}. Finally, the gadgets needed for the reductions to apply are given in Section~\ref{sec:cis} as well as the proofs of the main theorems.

\section{Definition of the Meta-Problem and Related Concepts}\label{sec:def}

\subsection{Preliminaries}
\subparagraph*{Notations.}
Throughout the paper, we consider simple non oriented graphs.

Given a positive integer $n$, the set of positive integers smaller than $n$ is denoted by $\intset{n}$. By extension, we define $\inftyset = \mathbb{N}_{>0} \cup \{\infty\}$.
Given two vertices $u,v$ of a graph $G$, the distance between $u$ and $v$ corresponds to the number of vertices in the shortest path between $u$ and $v$ and is denoted  $d(u,v)$. The \textit{open neighborhood of $u$} is denoted by $N(u)$, its \textit{closed neighborhood} is $N[u] = N(u) \cup \{u\}$, and for a value $r \in \inftyset$, the \textit{$r$-neighborhood} of $u$ is $N_r[u]$, that is the set of vertices at distance less than $r+1$ of $u$.
For $r = \infty$, the $\infty$-neighborhood of $u$ is the set of vertices in the same connected component than $u$.
We recall that a subset $D$ of $V$ is called an \textit{$r$-dominating set} of $G$ if for all vertices $u$ of $V$, the set $N_r[u] \cap D$ is non-empty. Thus an $\infty$-dominating set of $G$ contains at least a vertex for each connected component of $G$.

Given two subsets $X$ and $Y$ of $V$, the distance $d(X,Y)$ corresponds to the value $d(X,Y) = \min \{d(x,y) ~ | ~ x \in X, ~y \in Y\}$.
For a vertex $u$, we will also use $d(u,X)$ and $d(X,u)$, defined similarly.
The \textit{symmetric difference} between $X$ and $Y$ is denoted by $X \, \Delta \,  Y$, and the \textit{$2$-combination} of a set $X$ is denoted $\mathcal{P}_2(X)$

Given two graphs $G = (V_G,E_G)$ and $H = (V_H,E_H)$, $H$ is an \textit{induced subgraph of $G$} if $V_H \subseteq V_G$ and for all vertices $u$ and $v$ of $V_H$, $(u,v) \in E_G$ if and only if $(u,v) \in E_H$. We denote $H = G[V_H]$ and $V_G \setminus V_H$ by $V_{G \setminus H}$. Symmetrically, $G$ is an \textit{induced supergraph of $H$}.

\paragraph*{The {\sc (Planar)} \HS{} problem.}\label{sub:hs}

Consider a universe of $n$ elements denoted $\Omega = \{u_i \; | \; i \in \intset{n}\}$ and a set of $m$ non-empty subsets of $\Omega$ denoted $\mathcal{S} = \{S_i \; | \; i \in \intset{m}\}$ such that every element belongs to at least a subset. Then, a subset of $\Omega$ intersecting every set of $\mathcal{S}$ is called an \textit{hitting set} of $\mathcal{S}$:\medskip
\medskip
\begin{qproblem}
 {\HS{}}
 {A universe $\Omega$ and a set $\mathcal{S}$ of non-empty subsets of $\Omega$ whose union covers $\Omega$.}
 {A minimal-sized hitting set $C$ of $\mathcal{S}$, \small{i.e. a subset of $\Omega$ satisfying $\forall S_i \in \mathcal{S}$, $S_i \cap C \neq \emptyset$.}}
\end{qproblem}
\noindent The parameterized version \HS{(k)} decides if there exists a hitting set of size $k$.

\begin{theorem}[R.G. Downey and M.R Fellows \cite{DowneyF99}] \label{thm:hs}
 \HS{} cannot be solved in $2^{o(n)}$ time under ETH even if $m = \mathcal{O}(n)$. Moreover, \HS{(k)} is $\W{2}$-hard.
\end{theorem}

\HS{} may be translated into a dominating problem on bipartite graphs. Given an instance $(\Omega, \mathcal{S})$ of \HS{}, let us define $\phi(\Omega, \mathcal{S}) = (V_{\Omega} \cup V_{\mathcal{S}},E)$ as the bipartite graph of size $n+m$ such that for each $i \in \intset{n}$, there exists a vertex $v^{\Omega}_i$ in $V_{\Omega}$, for each $j \in \intset{m}$, there exists a vertex $v^{\mathcal{S}}_j$ in $V_{\mathcal{S}}$, and the edge $(v^{\Omega}_i,v^{\mathcal{S}}_j)$ is present in $E$ if and only if the element $u_i$ belongs to the subset $S_j$. Henceforth, a hitting set of $ \mathcal{S}$ is equivalent to a subset $C$ of $V_{\Omega}$ that dominates $V_{\mathcal{S}}$. We call $\phi(\Omega, \mathcal{S})$ \textit{the associated graph of $(\Omega, \mathcal{S})$}.\medskip

\begin{qproblem}{\textsc{Planar} \HS{}}
 {An instance $(\Omega, \mathcal{S})$ of \HS{} such that $\phi(\Omega, \mathcal{S})$ is planar.}
 {A hitting set $C$ of $\mathcal{S}$ of minimal size.}
\end{qproblem}
\noindent We also consider the parameterized version \textsc{Planar} \HS{(k)} of the latter problem.

\begin{theorem}[folklore]
 \label{thm:planarHS}
 There exists a reduction from \textsc{SAT} to \textsc{Planar} \HS{(n)} \linebreak producing associated graphs of quadratic size in the number $n$ of variables of the instances of \textsc{SAT}. Thus \textsc{Planar} \HS{} cannot be solved in $2^{o(\sqrt{n})}$ under ETH even if $m = \mathcal{O}(n)$.
\end{theorem}

\begin{proof}
 Let $\Phi$ be the set of the $n$ variables present in the set $\mathcal{C}$ of clauses of an instance of \textsc{SAT}. For each variable $\varphi$ of $\Phi$, we add two fresh elements $u_{\varphi}$ and $\bar{u}_{\varphi}$ to the universe $\Omega_{\Phi}$ representing the two possible affectations of variable $\varphi$, and we create a set $S_{\varphi} = \{u_{\varphi}, \bar{u}_{\varphi}\}$ that we append to the set $\mathcal{S}_{\mathcal{C}}$ of subsets of $\Omega_{\Phi}$. The independence of the sets $S_{\varphi}$ implies that the existence a hitting set of size strictly smaller than $n$ is impossible. Reciprocally, a potential hitting set of size exactly $n$ must define an affectation of the $n$ variables of $\Phi$. Finally, to determine if an affectation satisfies the set of clauses $\mathcal{C}$, for each clause $c \in \mathcal{C}$ we append to $\mathcal{S}_{\mathcal{C}}$ the set of elements representing each literal present in the clause $c$. The equivalence between the satisfiability of $\mathcal{C}$ and the existence of a hitting set of $\mathcal{S}_{\mathcal{C}}$ of size $n$ is immediate by construction.
 It remains to guarantee the planarity of the associated graph $\phi(\Omega_{\Phi},\mathcal{S}_{\mathcal{C}})$ . To do so, we actually apply the reduction on a restriction of \textsc{SAT} named \textsc{Separate Simple Planar SAT} (See \cite{tippenhauer} for a precise definition). Adding the \textit{sparsifying lemma} from \cite{ImpagliazzoPZ01}, the reduction produces a graph of size linear in $n$, preserving the computational lower bound of \textsc{Separate Simple Planar SAT}. In particular, the latter problem is not solvable in $2^{o(\sqrt{n})}$ under ETH,
\end{proof}

\subsection{The \DIS{} meta-problem}

Given a graph $G = (V,E)$ and $r \in \inftyset$, the classical identifying sets may be rewritten:
\begin{itemize}
 \item $r$-\aIC{}: a subset $C$ of $V$ is a $r$-identifying code of $G$ if it is an $r$-dominating set and for every distinct vertices $u, v$ of $V$, a vertex $w$ in $C$ verifies $w \in N_r[u] \, \Delta \, N_r[v]$.
       
 \item $r$-\aLD{}: a subset $C$ of $V$ is a $r$-locating dominating set of $G$ if it is an $r$-dominating set and for every distinct vertices $u, v$ of $V$, a vertex $w$ in $C$ verifies $w \in (N_r[u] \, \Delta \, N_r[v]) \cup \{u,v\}$.

 \item $r$-\aMD{}: a subset $C$ of $V$ is a $r$-resolving set of $G$ if it is an $r$-dominating set and for every distinct vertices $u, v$ of $V$, a vertex $w$ in $C$ verifies $w \in N_r[u] \cup N_r[v]$ and $d(u,w) \neq d(v,w)$.
\end{itemize}

A pattern clearly appears: the previous identifying sets only deviate on the criterion that the vertex $w$ must verify. The pivotal idea is to consider an abstract version of the criterion which does not depend on the input graph. Hence:

\begin{definition}[identifying function] A function $f$ of type: $G \to (V \times \mathcal{P}_2(V) \to \{ \true, \false\})$, is called \textit{an identifying function}. Given three vertices $u$, $v$ and $w$ of a graph $G$ such that $u \neq v$, we write $f_G[w](u,v)$ to get the resulting boolean. The notation $\mathcal{P}_2(V)$ implies that $f_G$ is symmetric, that is $f_G[w](u,v) = f_G[w](v,u)$.
\end{definition}

We need to require some useful properties on identifying functions to produce generic results. By mimicking the classical identifying sets, the main property we consider is that a vertex cannot distinguish two vertices at the same distance from it. Then:

\begin{definition}[distance function] \label{def:dif}
 A \textit{distance identifying function $f$} is an identifying function such that for every graph $G$ and all vertices $u$,$v$ and $w$ of $G$ with $u \neq v$: \begin{itemize}
  \item[($\alpha$)] \label{pdistance} $f_G[w](u,v)$ is $\false$ when $d(u,w) = d(v,w)$.
 \end{itemize}
\end{definition}

\noindent Besides this mandatory criterion, we suggest two paradigms related to the neighborhood of a vertex. Let $i \in \inftyset$. First, we may restrain the range of a vertex to its $i$-neighborhood: a vertex should not distinguish two vertices if they do not lie in its $i$-neighborhood but it should always distinguish them whenever exactly one of them lies to that $i$-neighborhood. Reciprocally, we may ensure that a vertex could distinguish the vertices of its $i$-neighborhood: a vertex should distinguish a vertex belonging to its $i$-neighborhood from all the other vertices, assuming the distances are different. Formally, we have:

\begin{definition}[$i$-local function] \label{def:lif}
 For $i \in \inftyset$, an \textit{$i$-local identifying function $f$} is an identifying function such that for every graph $G$ and all vertices $u$, $v$, $w$ of $G$ with $u \neq v$: \begin{enumerate}
  \item[($\beta_1$)] \label{plocal1} $f_G[w](u,v)$ is $\true$ when $d(u,w) \leq i < d(v,w)$ or, symmetrically, $d(v,w) \leq i < d(u,w)$.
  \item[($\beta_2$)] \label{plocal2} $f_G[w](u,v)$ is $\false$ when $i < \min\{d(u,w), d(v,w)\}$.
 \end{enumerate}
\end{definition}

\begin{definition}[$i$-\layered{} function] \label{def:cif}
 For $i \in \inftyset$, an \textit{$i$-\layered{} identifying function $f$} is an identifying function such that for every graph $G$ and all vertices $u$,$v$,$w$ of $G$ with $u \neq v$: \begin{enumerate}
  \item[($\gamma$)] \label{playered} $f_G[w](u,v)$ is $\true$ when $\min \{d(u,w), d(v,w)\} \leq i$ and $d(u,w) \neq d(v,w)$.
 \end{enumerate}
\end{definition}

In the following, given an identifying function $f$ and three vertices $u$, $v$, $w$ of a graph $G$, we say that $w$ $f$-distinguishes $u$ and $v$ if and only if $f_G[w](u,v)$ is $\true$.
By extension, given three vertex sets $C$, $X$ and $Y$ of $G$, we say that $C$ $f$-distinguishes $X$ and $Y$ if for every $u$ in $X$ and $v$ in $Y$, either $u = v$ or there exists $w$ in $C$ verifying $f_G[w](u,v)$. Finally, a graph $G$ of vertex set $V$ is $f$-distinguished by $C$ when $C$ $f$-distinguishes $V$ and $V$.

We are now ready to define the \DIS{} meta-problem.
\begin{definition}[$(f,r)$-distance identifying set]
 For a distance identifying function $f$ and $r \in \inftyset$, a \textit{$(f,r)$-distance identifying set} of a graph $G$ is an $r$-dominating set of $G$ that $f$-distinguishes $G$.
\end{definition}

\begin{qproblem}{\DIS{}}
 {A distance identifying function $f$ and $r \in \inftyset$. A graph $G$.}
 {A $(f,r)$-distance identifying set of $G$ of minimal size, if one exists.\medskip}
\end{qproblem}

Given a distance identifying function $f$ and $r \in \inftyset$ as inputs of the meta-problem, the resulting problem is called $(f,r)$-\DIS{} and denoted $(f,r)$-\aDIS{}. The problem $(f,r)$-\aDIS{} is said to be \textit{$i$-\layered{}} when the function $f$ is $i$-\layered{}, and it is said to be \textit{$i$-local} when $f$ is $i$-local and $r = i$. A problem is \textit{local} if it is $i$-local for an integer $i$.
Our local reductions will need a \textit{local gadget} to operate: the subfamilies of local problems admitting a (bipartite) local gadget is called \textit{(bipartite) gifted local}. We do not need to define \textit{gifted $1$-\layered{}} as every $1$-\layered{} problem admits a $1$-\layered{} gadget.
We also consider the parameterized version \DIS{($k$)}.

\subsection{Detailed Computational Lower Bounds}

Using the \DIS{} meta-problem, we get the following lower bounds:

\begin{theorem}\label{thm:CPlowerbound}
 For each $1$-\layered{} distance identifying function $f$ and every $r \in \inftyset{}$, the $(f,r)$-\DIS{} problem restricted to bipartite apex graphs is \NP-hard, and does not admit an algorithm running in $2^{\mathcal{O}(\sqrt{n})}$ time under ETH.
\end{theorem}
\begin{theorem}\label{thm:LPlowerbound}
 The (bipartite) gifted local problems restricted to (bipartite) planar graphs are \NP-hard, and do not admit an algorithm running in $2^{\mathcal{O}(\sqrt{n})}$ time under ETH.
\end{theorem}
\begin{theorem}\label{thm:LCPlowerbound}
 For each $r$-local $0$-\layered{} distance identifying function $f$, $(f,r)$-\aDIS{} restricted to bipartite planar graphs is \NP-hard, and cannot be solved in $2^{\mathcal{O}(\sqrt{n})}$ under ETH.
\end{theorem}

\begin{theorem}\label{thm:Glowerbound}
 Let $f,g$ and $h$ be distance identifying functions such that $f$ is $1$-\layered{}, $g$ is $q$-local $0$-\layered{} and $h$ is $p$-local and admits a local (bipartite) gadget. Let $r \in \inftyset{}$.
 The $(f,r)$-, $(g,q)$- and $(h,p)$-\aDIS{} problems are \NP-hard, and do not admit:\begin{itemize}
  \item algorithms running in $2^{o(n)}$ time, except if ETH fails,
  \item parameterized algorithms running in $2^{\mathcal{O}(k)} \cdot n^{\mathcal{O}(1)}$ time, except if $W[2] = \FPT{}$.
 \end{itemize}
 \noindent The parameter $k$ denotes here the solution size of a relevant distance identifying set.
 
 \noindent All bounds still hold in the bipartite case (whenever the gadget associated with $h$ is bipartite).
\end{theorem}

As a side result, the $1$-\layered{} general reduction answers a question of Hartung in \cite{HartungN13}:\begin{corollary}
 Under ETH, \MD{} cannot be solved in $2^{o(n)}$.
\end{corollary}

Finally, notice that the parameterized lower bound from Theorem~\ref{thm:Glowerbound} may be complemented by an elementary upper bound inspired from the kernel of \aIC{} and \aLD{} of size $2^k+k$:
\begin{proposition} \label{prop:upperbound}
 For every $r$-local distance identifying function $f$, the $(f,r)$-\textsc{Distance Identifying Set} problem has a kernel of size $(r+1)^{k}+k$ where k is the solution size. Therefore, it admits a naive parameterized algorithm running in $\mathcal{O}(n^{k+3}) \in \mathcal{O^*}(r^{(k^2)})$ time.
\end{proposition}
\begin{proof} The kernel size simply relies on the fact that $k$ vertices may characterize at most $(r+1)^k$ $r$-neighbors using distances, while the parameterized algorithm just enumerates the $\binom{n}{k}$ set of $k$ vertices of the input graph, trying them in $\mathcal{O}(n^3)$.
\end{proof}

The proofs of the Theorems~\ref{thm:CPlowerbound} to~\ref{thm:Glowerbound} will be given in Section~\ref{sec:cis}.

\section{The Supports of the Reductions for \DIS{}}\label{sec:sup}
\subsection{The Distance Identifying Graphs}

Consider the associated graph $\phi(\Omega,\mathcal{S})$ as defined in Section~\ref{sub:hs}. The differences between the \DIS{} meta-problem and the dominating problem related to associated graphs actually raise two issues for a reduction based on these latter notions to be effective on \DIS{}. First, contrarily to the dominating problem where a vertex may only discern its close neighborhood, the meta-problem may allow a vertex to discern further than its direct neighborhood. In that case, we cannot certify that a vertex $v^{\Omega}_i$ does not distinguish a vertex $v^{\mathcal{S}}_j$ when $u_i$ is not in $S_j$, the adjacency not remaining a sufficient argument. Secondly, one may object that a vertex $v^{\Omega}_i$ formally has to distinguish a vertex $v^{\mathcal{S}}_j$ from another vertex, but that distinguishing a single vertex is not defined.

To circumvent these problems, we suggest the following fix: rather than producing a single vertex for each $S_j \in \mathcal{S}$, the set $V_\mathcal{S}$ may contain two vertices $v^{\mathcal{S}}_j$ and $\bar{v}^{\mathcal{S}}_j$. Then, the role of $v^{\Omega}_i$ would be to distinguish them if and only if $u_i \in S_j$.
To ensure that the vertex $v^{\Omega}_i$ distinguishes $v^{\mathcal{S}}_j$ and $\bar{v}^{\mathcal{S}}_j$ when $u_i \in S_j$, we may use the properties $(\beta_1)$ and $(\gamma)$ of Definition~\ref{def:lif} and~\ref{def:cif} for the $r$-local and $1$-\layered{} problems, respectively. Precisely, when $u_i \in S_j$, $v^{\Omega}_i$ should be at distance $r$ to $v^{\mathcal{S}}_j$ (with $r = 1$ in the $1$-\layered{} cases) while $\bar{v}^{\mathcal{S}}_j$ should not be in the $r$-neighborhood of $v^{\Omega}_i$.
Similarly, to ensure that $v^{\Omega}_i$ cannot distinguish $v^{\mathcal{S}}_j$ and $\bar{v}^{\mathcal{S}}_j$ when $u_i \not \in S_j$, we may use properties $(\alpha)$ or $(\beta_2)$ of Definitions~\ref{def:dif} and~\ref{def:lif}. Hence, when $u_i \not \in S_j$, $v^{\mathcal{S}}_j$ should not be in the $r$-neighborhood of $v^{\Omega}_i$, or $d(v^{\Omega}_i,v^{\mathcal{S}}_j)$ and $d(v^{\Omega}_i,\bar{v}^{\mathcal{S}}_j)$ should be equal.

That fix fairly indicates how to initiate the transformation of the associated graphs in order to deliver an equivalence between a hitting set formed by elements of $\Omega$ and the vertices of a distance identifying set included in $V_\Omega$. However, it is clearly not sufficient since we also have to distinguish the couples of vertices of $V_\Omega$ for which nothing is required. To solve that problem, we suggest to append to each vertex of the associated graph a copy of some gadget with the intuitive requirement that the gadget is able to distinguish the close neighborhood of its vertices from the whole graph. We introduce the notion of \textit{$B$-extension}:

\begin{definition}[$B$-extension]
 Let $H = (V_H,E_H)$ be a connected graph, and $B \subseteq V_H$. An induced supergraph $G = (V_G,E_G)$ is said to be a \textit{$B$-extension} of $H$ if it is connected and for every vertex $v$ of $V_{G \setminus H}$, the set $N(v) \cap V_H$ is either equal to $\emptyset$ or $B$.
 
 A vertex $v$ of $V_{G \setminus H}$ such that $N(v) \cap V_H = B$ is said to be \textit{$B$-adjacent}. The $B$-extensions of $H$ such that $V_{G \setminus H}$ contains exactly a $B$-adjacent vertex or two $B$-adjacent vertices but not connected to each other are called \textit{the $B$-single-extension} and \textit{the $B$-twin-extension} of $H$, respectively.
\end{definition}

Here, the "border" $B$ makes explicit the connections between a copy of a gadget $H$ and a vertex outside the copy. In particular, a $B$-single-extension is formed by a gadget with its related vertex $v^{\Omega}_i$, while a $B$-twin-extension contains a gadget with its two related vertices $v^{\mathcal{S}}_j$ and $\bar{v}^{\mathcal{S}}_j$.
Piecing all together, we may adapt the associated graphs to the meta-problem:

\begin{definition}[$(H,B,r)$-distance identifying graph]
 Let $(\Omega = \{u_i \; | \; i \in \intset{n}\},\mathcal{S} = \{S_i \; | \; i \in \intset{m}\})$ be an instance of \HS{}. Let $H$ be a connected graph, $B$ a subset of its vertices, and $r$ a positive integer. The \textit{$(H,B,r)$-distance identifying graph $\Phi[H,B,r](\Omega,\mathcal{S})$} is as follows.
 \begin{itemize}
  \item for each $i \in \intset{n}$, the graph $\Phi[H,B,r](\Omega,\mathcal{S})$ contains as induced subgraph a copy $H^{\Omega}_i$ of $H$ together with a $B^{\Omega}_i$-adjacent vertex $v^{\Omega}_i$, where $B^{\Omega}_i$ denotes the copy of $B$.
  \item similarly, for each $j \in \intset{m}$, the graph $\Phi[H,B,r](\Omega,\mathcal{S})$ contains a copy $H^{\mathcal{S}}_j$ of $H$ together with two $B^{\mathcal{S}}_j$-adjacent vertices $v^{\mathcal{S}}_j$ and $\bar{v}^{\mathcal{S}}_j$ (the latter vertices are not adjacent) and where $B^{\mathcal{S}}_j$ denotes the copy of $B$.
  \item finally, for each $S_j \in \mathcal{S}$ and each $u_i \in S_j$, $v^{\Omega}_i$ is connected to $v^{\mathcal{S}}_j$ by a path of $r-1$ vertices denoted $l^k_{i,j}$ with $d(v^{\Omega}_i,l^k_{i,j}) = k$ for each $k \in \intset{r-1}$.
 \end{itemize}
\end{definition}

When the problem is not local, we prefer the following identifying graph:
\begin{definition}[$(H,B)$-apex distance identifying graph]
 An \textit{$(H,B)$-apex distance identifying graph $\Phi^*[H,B](\Omega,\mathcal{S})$}  is the union of a $(H,B,1)$-distance identifying graph with an additional vertex $a$ called \textit{apex} such that:
 \begin{itemize}
  \item for each $u_i \in \Omega$, the apex $a$ is $B^\Omega_i$-adjacent to $H^\Omega_i$, where $B^{\Omega}_i$ (resp.  $H^\Omega_i$) denotes the copy of $B$ (resp. $H$).
  \item for each $S_j \in \mathcal{S}$, the apex $a$ is adjacent to $v^{\mathcal{S}}_j$ and $\bar{v}^{\mathcal{S}}_j$.
 \end{itemize}
\end{definition}
See Figure~\ref{fig:graphidentifying} for an example of an $(H,B,r)$-distance identifying graph (on the left) and an example of $(H,B)$-apex distance identifying graph (on the right).
\begin{proposition} \label{prop:graphsize}
 Given an instance $(\Omega,\mathcal{S})$ of \textsc{Planar} \HS{} where $|\Omega| = n$, $|\mathcal{S}| = m$, the graphs $G = \Phi[H,B,r](\Omega,\mathcal{S})$ and  $G' = \Phi^*[H,B](\Omega,\mathcal{S})$\begin{itemize}
  \item are connected and have size bounded by $(|H|+2r)(n+m)$, (with $r = 1$ for $G'$),
  \item may be built in polynomial time in their size,
  \item are bipartite if the $B$-single extension of $H$ is bipartite,
  \item are respectively planar and an apex graph if the $B$-twin-extension of $H$ is planar.
        
 \end{itemize}
\end{proposition}

\begin{proof}
 The graph $G$ is formed by the union of $n$ $B$-single-extensions of $H$, $m$ $B$-twin-extensions of $H$ and the all the possible paths of $r-1$ vertices.
 As $\phi(\Omega,\mathcal{S})$ is a bipartite planar graph, the Euler formula implies that the number of paths is bounded by $2(n+m)-4$. We conclude that the number of vertices of $G$ is bounded by: \[n(|H|+1)+m(|H|+2)+(r-1)(2(n+m)-4) = (|H|+2r)(n+m)-n-4(r-1)\]
 Furthermore, it is clear that $G$ is connected if and only if the associated graph $\phi(\Omega,\mathcal{S})$ is connected. Additionally, we may consider that $\phi(\Omega,\mathcal{S})$ is connected since it is a property decidable in polynomial time, and that the instances corresponding to the distinct connected components of $\phi(\Omega,\mathcal{S})$ may be considered independently.\medskip
 
 Finally, all the other items of the proposition are direct by construction.
\end{proof}

\begin{figure}[t]%
 \centering
 \begin{minipage}[b]{0.47\linewidth}
  \centering
  \includegraphics[width=.8\textwidth]{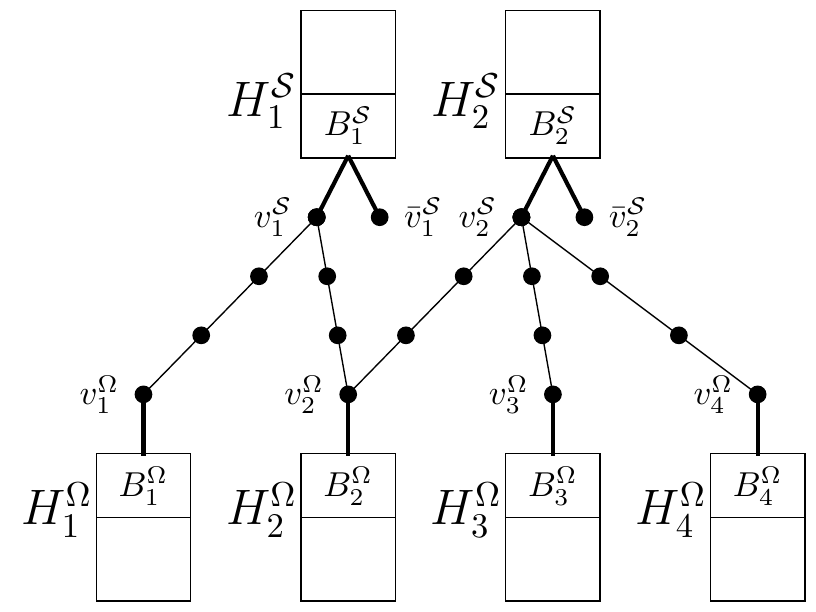}
 \end{minipage}
 \hspace{0.5cm}
 \begin{minipage}[b]{0.47\linewidth}
  \centering
  \includegraphics[width=.8\textwidth]{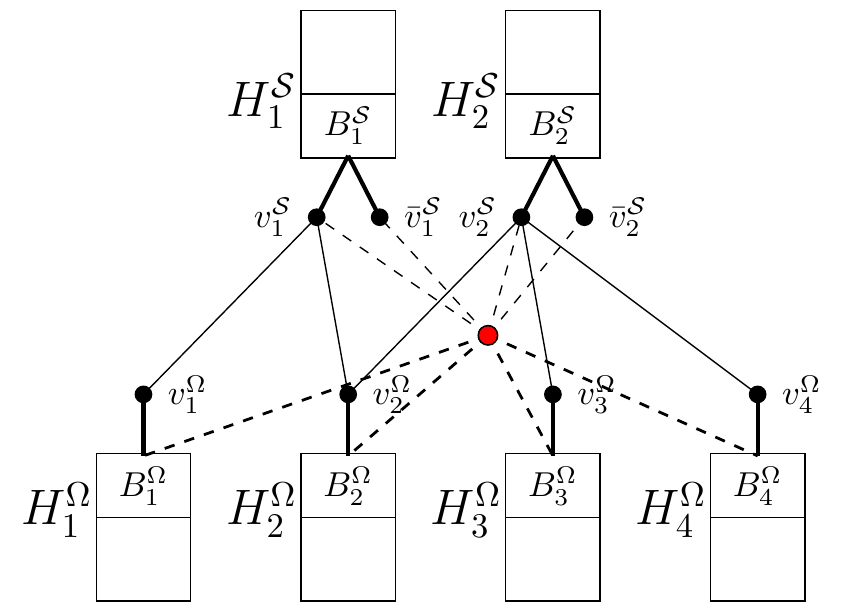}
 \end{minipage}
 
 \caption{A $(H,B,3)$-distance identifying graph and a $(H,B)$-apex distance identifying graph built on the planar instance formed by $\Omega = \{1,2,3,4\}$ and $\mathcal{S} = \{\{1,2\},\{2,3,4\}\}$.}
 \label{fig:graphidentifying}
\end{figure}

Having defined the (apex) distance identifying graphs, the main effort to obtain generic reduction from \textsc{Planar} \HS{} is done. We now define relevant gadgets:

\begin{definition}[$(f,r)$-gadgets] \label{def:gadgets}
 Let $f$ be a distance identifying function and $r \in \inftyset$.
 Let $H = (V_H,E_H)$ be a connected graph, and $B,C$ be two subsets of $V_H$.
 We said that the triple $(H,B,C)$ is a \textit{$(f,r)$-gadget} if for every $B$-extension $G$ of $H$:
 \begin{itemize}
  \item[$(p_h)$] $C$ $f$-distinguishes $V_H$ and $V_G$.
  \item[$(p_b)$] $C$ $f$-distinguishes $N_B$ and $V_{G \setminus H} \setminus N_B$, where $N_B$ is the set of $B$-adjacent vertices of $G$.
  \item[$(p_d)$] $C$ is an $r$-dominating set of $G[V_H \cup N_B]$.
  \item[$(p_s)$] For all $(f,r)$-distance identifying set $S$ of $G$, $|C| \leq |S \cap V_H|$.
 \end{itemize}
\end{definition}

\begin{definition}[local gadgets]
 \label{def:localgadgets}
 A $(f,r)$-gadget is a \textit{local gadget}, if $f$ is a $r$-local identifying function with $r \neq \infty$, and
 \begin{itemize}
  \item[$(p_l)$] for every $k \in \intset{r}$, there exists $c \in C$ such that $d(c,B) = k-1$.
 \end{itemize}
\end{definition}

Consistently, we say that a $(f,r)$-gadget $(H,B,C)$ is \textit{bipartite} if the $B$-single-extension of $H$ is bipartite, and that it is \textit{planar} if the $B$-twin-extension of $H$ is planar.

\begin{theorem}\label{thm:setsize}
 Let $(\Omega,\mathcal{S})$ be an instance of \HS{} such that $|\Omega| = n > 1$, $|\mathcal{S}| = m$. Let $(H,B,C)$ be a $(f,r)$-gadget for a $1$-\layered{} identifying function $f$ and let $(H',B',C')$ be a local $(g,q)$-gadget. The following propositions are equivalent: \begin{itemize}
  \item there exists a hitting set of $\mathcal{S}$ of size $k$.
  \item there exists a $(f,r)$-distance identifying set of $\Phi^*[H,B](\Omega,\mathcal{S})$ of size $k+|C|(n+m)$.
  \item there exists a $(g,q)$-distance identifying set of $\Phi[H',B',q](\Omega,\mathcal{S})$ of size $k+|C'|(n+m)$.
 \end{itemize}
\end{theorem}
\begin{proof} We start by focusing on the equivalence between the first and second items.
 
 Suppose first that $P$ is a hitting set of $(\Omega,\Set)$ of size $k$. By denoting $C^\Omega_i$ and $C^\Set_j$ the copies of $C$ associated to the copies $H^\Omega_i$ and $H^\Set_j$ of $H$, we suggest the following set $I$ of size $k+|C|(n+m)$ as a $(f,r)$-distance identifying set of $G = \Phi^*[H,B](\Omega,\Set)$:
 \[I = \{ v^\Omega_i \colon u_i \in P \} \; \cup \underset{i \, \in \, \intset{n}}{\bigcup} C^\Omega_i \cup \underset{j \, \in \, \intset{m}}{\bigcup}  C^\Set_j.\]
 
 Recall that by construction, $G$ is a $B^\Omega_i$-extension of $H^\Omega_i$ (respectively $B^\Set_j$-extension of $H^\Set_j$) for any $i \in \intset{n}$ (respectively $j \in \intset{m}$).
 This directly implies that $I$ is an $r$-dominating set of $G$. Indeed, the condition $(p_d)$ of Definition~\ref{def:gadgets} implies that $C^\Omega_i$ (respectively $C^\Set_j$) $r$-dominates $H^\Omega_i$ plus $v^\Omega_i$ (respectively of $H^\Set_j$ plus $v^\Set_j$, $\ov^\Set_j$). The remaining apex is also $r$-dominated by any $C^\Omega_i$, as it is $B^\Omega_i$-adjacent for every $i \in \intset{n}$.
 
 We now have to show that $I$ $f$-distinguishes $G$. We begin with the vertices of the gadget copies because the condition $(p_h)$ implies that $C_i^\Omega \subseteq I$ $f$-distinguishes the vertices of $H_i^\Omega$ and $G$ for every $i \in \intset{n}$, and $I$ $f$-distinguishes the vertices of $H_j^\Set$ and $G$ for every $j \in \intset{m}$.
 Thereby, we only have to study the vertices of the form $v^\Omega_i$, $v^\Set_j$, $\ov^\Set_j$, and the apex $a$ (there is no vertex of the form $l^k_{i,j}$ in an apex distance identifying graph).
 To distinguish them, we use the condition $(p_b)$. Recall that $n > 1$. Then, for each distinct $i,i' \in \intset{n}$, we have:
 \begin{itemize}
  \item $v^\Omega_i$ is $B^\Omega_i$-adjacent but not $B^\Omega_{i'}$-adjacent,
  \item $a$ is both $B^\Omega_i$-adjacent and $B^\Omega_{i'}$-adjacent,
  \item a vertex of the form $v^\Set_j$ or $\ov^\Set_j$ is neither $B^\Omega_i$-adjacent nor $B^\Omega_{i'}$-adjacent.
 \end{itemize}
 Enumerating the relevant $i$ and $i'$, we deduce that every couple of vertices is distinguished except when they are both of the form $v^\Set_j$ or $\ov^\Set_{j'}$ for $j,j' \in \intset{m}$.
 But we may distinguish $v^\Set_j$ or $\ov^\Set_{j'}$ for distinct $j,j'$ by applying $(p_b)$ on $H^\Set_j$.
 
 It remains to distinguish $v^\Set_j$ and $\ov^\Set_{j}$ for $j \in \intset{m}$.
 We now use the fact that $P$ is a hitting set for $(\Omega,\Set)$.
 By definition of a hitting set, for any set $S_j \in \Set$, there exists a vertex $u_i \in P$ such that $u_i \in S_j$. We observe that $d(v^\Omega_i,v^\Set_j) = 1 < d(v^\Omega_i,\ov^\Set_j)$ by construction of $G$ and that $v^\Omega_i \in I$ by definition of $I$. Since $f$ is $1$-\layered{}, $I$ $f$-distinguishes $v^\Set_j$ and $\ov^\Set_j$.\medskip
 
 In the other direction, assume that $I$ is a distance identifying set of $G$ of size $k+|C|(n+m)$.
 As every set of $S$ is not empty, we may define a function $\varphi : \intset{m} \to \intset{n}$ such that $u_{\varphi(j)} \in S_j$.
 
 We suggest the following set $P$ as an hitting set of $\mathcal{S}$ of size at most $k$:
 \[P = \{ u_i \in \Omega \; | \; v^\Omega_i \in I \} \cup \{ u_{\varphi(j)} \in \Omega \; | \; v^\Set_j \in I \text{ or } \ov^\Set_j \in I\}\]
 
 We claim that the only vertices that may $f$-distinguish $v^\Set_j$ and $\ov^\Set_j$ are themselves and the vertices $v^\Omega_i$ such that $u_i \in S_j$. To prove so, we apply propriety $(\alpha)$ of Definition~\ref{def:dif}: \begin{itemize}
  \item the apex $a$ verifies $d(a,v^\Set_j) = 1 = d(a,\ov^\Set_j)$
  \item a vertex $v^\Omega_i$ such that $u_i \not \in S_j$ verifies $d(v^\Omega_i,v^\Set_j) = 3 = d(v^\Omega_i,\ov^\Set_j)$
  \item a vertex $v$ of $H^\Omega_i$ verifies $d(v,v^\Set_j) = 2+d(v,B^\Omega_i) = d(v,\ov^\Set_j)$
  \item a vertex $v$ of $H^\Set_{j'}$ with $j \neq j'$ verifies $d(v,v^\Set_j) = 3+ d(v,B^\Set_j) = d(v,\ov^\Set_j)$
  \item both $v^\Set_j$ and $\ov^\Set_j$ are $B^\Set_j$-adjacent, so they are at the same distance of any vertex of $H^\Set_j$.
 \end{itemize}
 
 We deduce that $v^\Set_j$ and $\ov^\Set_j$ are $f$-distinguished only if either one on them belongs to $I$ (in that case $u_{\varphi(j)} \in P \cap S_j$) or there exists $v^\Omega_i \in I$ such that $u_i \in S_j$ (and then $u_i \in P \cap S_j$).
 
 It remains to show that $|P| \leq k$. By the condition $(p_s)$ of Definition~\ref{def:gadgets}, we know that $|I \cap V_{H^\Omega_i}| \geq |C^\Omega_i|$ and  $|I \cap V_{H^\Set_j}| \geq |C^\Set_j|$ for any $i \in \intset{n}$ and $j \in \intset{m}$, implying\par
 \begin{center}\makebox[\linewidth]{$k = |I| - |C|(n+m)  \geq \sum\limits_{i \in [n]} |I \cap \{v^\Omega_i\}| + \sum\limits_{j \in [m]} |I \cap \{v^\Set_j ,\ov^\Set_j\}| \geq \sum\limits_{v^\Omega_i \in I} 1 + \sum\limits_{ I \cap \{v^\Set_j,\ov^\Set_j\} \neq \emptyset} 1 = |P|$}
 \end{center}\par
 
 Now, we prove the equivalence between the first and third items.
 Consider a $r$-local distance identifying function $f$, a local $(f,r)$-gadget $(H,B,C)$ and an instance $(\Omega,\Set)$ of \textsc{Planar} \HS{} such that $|\Omega| = n$, $|\Set| = m$.
 We denote the copies of $H$ as $H_i^\Omega$ or $H_j^\Set$, the copies of $C$ as $C_i^\Omega$ and $C_j^\Set$, and the copies of $B$ as $B_i^\Omega$ and $B_j^\Set$  for any $i \in \intset{n}$ and $j \in \intset{m}$.
 
 In the first direction, suppose that $P$ is a hitting set of $(\Omega,\Set)$ of size $k$, the $(g,r)$-distance identifying set $I$ of $G = \Phi[H,B,r](\Omega,\Set)$ is defined identically as in the equivalence of the first and second items of the current theorem:
 \[ I = \{ v_i^\Omega \colon u_i \in P\} \cup \bigcup_{i \in [n]} C_i^\Omega \cup \bigcup _{j \in [m]} C_j^\Set\]
 Using conditions $(p_d)$ and $(p_l)$ of Definitions~\ref{def:gadgets} and~\ref{def:localgadgets}, $I$ is clearly an $r$-dominating set of $G$. Indeed, by $(p_d)$ every vertex belonging to a copy of the gadget is $r$-dominated. Additionally, every vertex outside of the copies of the gadgets is at distance at most $r$ of a copy by construction, but there exists a vertex $b \in B \cap C$ (so a relevant copy in $I$) by $(p_l)$.
 
 To prove that $I$ $f$-distinguishes $G$, the strategy is differing from the previous equivalence only on the $l_{i,j}^k$ vertices and when distinguishing $v^\Set_j$ and $\ov^\Set_j$ as we will see.
 
 Recall that by construction, $G$ is a $B^\Omega_i$-extension of $H^\Omega_i$ (respectively $B^\Set_j$-extension of $H^\Set_j$) for any $i \in \intset{n}$ (respectively $j \in \intset{m}$).
 Distinguishing the vertices of the gadget copies is easy, as the condition $(p_h)$ implies that $C_i^\Omega \subseteq I$ $f$-distinguishes the vertices of $H_i^\Omega$ and $G$ for every $i \in \intset{n}$, and similarly $I$ $f$-distinguishes the vertices of $H_j^\Set$ and $G$ for every $j \in \intset{m}$.

 Thereby, we only have to study the vertices of the form $v^\Omega_i$, $v^\Set_j$, $\ov^\Set_j$,  and the vertices $l_{i,j}^k$.

 To distinguish them, we mainly use the condition $(p_b)$.
 We observe that for each distinct $i,i' \in \intset{n}$ (they exist as $n > 1$) : \begin{itemize}
  \item $v^\Omega_i$ is $B^\Omega_i$-adjacent but not $B^\Omega_{i'}$-adjacent,
  \item for every $j \in [m]$, $v^\Set_j$ or $\ov^\Set_j$ is neither $B^\Omega_i$-adjacent nor $B^\Omega_{i'}$-adjacent.
  \item for every $i \in [n]$, $j \in [m]$ and $k \in [r-1]$, $l_{i,j}^k$ is neither $B^\Omega_i$-adjacent nor $B^\Omega_{i'}$-adjacent.
 \end{itemize}
 Thus $I$ $f$-distinguishes $v_i^\Omega$ and $G$.
 
 As the vertices of form $l_{i,j}^k$ are the only ones to belong to both the $r$-neighbourhood of $B_i^\Omega$ and $B_j^\Set$, and as the vertices $l_{i,j}^k$ and $l_{i,j}^{k'}$ with $k < k'$ are $f$-distinguished by the guaranteed vertex $c \in C_i^\Omega$ such that $d(c,B_i^\Omega) = r-k-1$, $I$ $f$-distinguishes $l_{i,j}^k$ and $G$ for every relevant $i,j$ and $k$.
 
 It remains to distinguish $v^\Set_j$ and $\ov^\Set_{j'}$ for $j,j' \in \intset{m}$. If $j$ and $j'$ are distinct we may use $(p_b)$ on the copy $H_j^\Set$ of the gadget $H$. We may assume $j = j'$.
 We now use the fact that $P$ is a hitting set for $(\Omega,\Set)$.
 By definition of an hitting set, for any set $S_j \in \Set$, there exists a vertex $u_i \in P$ such that $u_i \in S_j$.
 We observe that $d(v^\Omega_i,v^\Set_j) = r < d(v^\Omega_i,\ov^\Set_j)$ (when $u_i \in S_j$) by construction of $G$ and $I$ indeed $f$-distinguishes $v^\Set_j$ and $\ov^\Set_j$ because $f$ is $r$-local.\medskip
 
 In the other direction, assume that $I$ is a distance identifying set of $G$ of size $k+|C|(n+m)$.
 The hitting set $P$ may now depend on $l_{i,j}^k$. Let define $L_i = \{v^\Omega_i\} \cup \{l_{i,j}^k \; | \; k \in \intset{r-1} \text{ and } u_i \in S_j\}$.
 As every set of $S$ is not empty, we may define a function $\varphi : \intset{m} \to \intset{n}$ such that $u_{\varphi(j)} \in S_j$.
 We suggest the following set $P$ as an hitting set of $\mathcal{S}$ of size at most $k$:
 \[P = \{ u_i \in \Omega \; | \; I \cap L_i \neq \emptyset \} \cup \{ u_{\varphi(j)} \in \Omega \; | \; v^\Set_j \in I \text{ or } \ov^\Set_j \in I\}\]
 
 Consider $j \in \intset{m}$, let us show that the only vertices that may $f$-distinguish the couple $(v^\Set_j, \ov^\Set_j)$ are themselves and the vertices from $L_i$ (and not only $v_i^\Omega$) such that $u_i \in S_j$.
 Every vertex from $H_j^\Set$ is at the same distance to $v^\Set_j$ and $\ov^\Set_j$ and thus cannot $f$-distinguishes them because of the distance property $(\alpha)$.
 Every vertex not in $H_j^\Set$, not in $L_i$ for every $i \in \intset{n}$ such that $u_i \in S_j$ and different from $v^\Set_j$ and $\ov^\Set_j$ is at distance at least $r+1$ of the two latter vertices.
 Thus, because of the propriety $(\beta_2)$ of Definition~\ref{def:lif} (a vertex cannot distinguish two vertices outside of its $r$-neighbourhood) any of these vertices does not $f$-distinguish $v^\Set_j$ and $\ov^\Set_j$.
 We deduce that $v^\Set_j$ and $\ov^\Set_j$ are $f$-distinguished if and only if either one of them belongs to $I$ (in that case $u_{\varphi(j)} \in P \cap S_j$) or there exists $i \in \intset{n}$ such that $u_i \in S_j$ and $I \cap L_i \neq \emptyset$.
 
 The proof that $|P| \leq k$ is provided by $(p_s)$, we know that $|I \cap V_{H^\Omega_i}| \geq |C^\Omega_i|$ and  $|I \cap V_{H^\Set_j}| \geq |C^\Set_j|$ for any $i \in \intset{n}$ and $j \in \intset{m}$.
 Considering the following partition of $I$
 \[I = \Big(\sqcup_{i \in \intset{n}} (I \cap H_i^\Omega) \Big) \bigsqcup \Big(\sqcup_{j \in \intset{m}} (I \cap H_j^\Set) \Big) \bigsqcup \Big( \sqcup_{i \in \intset{n}} (I \cap L_i) \Big) \bigsqcup \Big( \sqcup_{j \in \intset{m}} (I \cap \{ v_j^\Set , \ov_j^\Set)  \} \Big)\]
 We get
 
 \medskip
 \begin{tabular}{rlllllll}
  $|I|$ & $\ge |C|(n+m)$                                                             &                  
  +     & ${\displaystyle\sum_{i \in \intset{n}} |I \cap L_i|}$                      &                  
  +     & ${\displaystyle\sum_{j \in \intset{m}} |I \cap \{v^\Set_j ,\ov^\Set_j\}|}$                    
  \\
        & $\ge |C|(n+m)$                                                             &                  
  +     & ${\displaystyle\sum_{I \cap L_i \not= \emptyset} 1}$                       &                  
  +     & ${\displaystyle\sum_{ I \cap \{v^\Set_j,\ov^\Set_j\} \neq \emptyset} 1}$                      
        & =                                                                          & $|C|(n+m) + |P|$ 
 \end{tabular}
 \medskip
 
 Because $|I| = k + |C|(n+m)$, we conclude that $|P| \leq k$.
 
 Obviously, the second and third items are equivalent since they are both equivalent to the first item, which concludes the proof.
\end{proof}

\subsection{Binary Compression of Gadgets}

The Theorem~\ref{thm:setsize} is a powerful tool to get reductions, in particular in the planar cases. However, the number of involved gadgets does not allow to use Theorem~\ref{thm:w2bound}. This limitation is due to the  uses of a gadget per vertex to identify in the distance identifying graphs. Using power set, we may obtain a better order: given $k$ gadgets, we may identify $2^k-1$ vertices (we avoid to identify a vertex with the empty subset of gadgets). Thus, we will consider \textit{binary representations} of integers as sequences of bits, with weakest bit at last position. For a positive integer $n$, we define the integer $\len{n} =  1+\lfloor \log_2(n) \rfloor$ and introduce a new graph:

\begin{definition}[$(H,B,r)$-compressed graph]
 Let $(\Omega = \{u_i \; | \; i \in \intset{n}\},\mathcal{S} = \{S_i \; | \; i \in \intset{m}\})$ be an instance of \HS{}. Let $H$ be a connected graph, $B$ be a subset of its vertices, and $r$ be a positive integer. The \textit{$(H,B,r)$-compressed graph $\Psi[H,B,r](\Omega,\mathcal{S})$} is defined as follows. $\Psi[H,B,r](\Omega,\mathcal{S})$ contains as induced subgraphs $\len{n+1}$ copies of $H$ denoted $H^{\Omega}_i$ for $i \in \intset{\len{n+1}}$ and $\len{m}$ other copies of $H$ denoted $H^{\mathcal{S}}_j$ for $j \in \intset{\len{m}}$. Then:
 
 \begin{itemize}
  \item for each $j \in \intset{m}$, we add two non-adjacent vertices $v^{\mathcal{S}}_j$ and $\bar{v}^{\mathcal{S}}_j$. They are $B^{\mathcal{S}}_k$-adjacent for each $k \in \intset{\len{m}}$ such that the $k^{th}$ bit of the binary representation of $j$ is $1$.
  \item for each $i \in \intset{n}$, we add $r$ vertices denoted $l_i^{j-1}$ with $j \in \intset{r}$ to form a fresh path such that $d(v^{\Omega}_i,l_i^{j-1}) = j-1$ where $v^{\Omega}_i = l_i^0$. We make $v^{\Omega}_i$ $B^{\Omega}_k$-adjacent for each $k \in \intset{\len{n+1}}$ such that the $k^{th}$ bit of the binary representation of $i$ is $1$.
  \item for each $S_j \in \mathcal{S}$ and each $u_i \in S_j$, we add the edge $(l_i^{r-1},v^{\mathcal{S}}_j)$.
  \item we add $r$ vertices denoted $a^{j-1}$ with $j \in \intset{r}$ to form a path such that $d(a^0,a^{j-1}) = j-1$. The vertex $a^0$ is $B^{\Omega}_k$-adjacent for every $k \in \intset{\len{n+1}}$, and we add the edges $(a^{r-1},v^{\mathcal{S}}_j)$ and $(a^{r-1},\bar{v}^{\mathcal{S}}_j)$ for each $j \in \intset{m}$.
 \end{itemize}
\end{definition}

\noindent By definition of $\len{n+1}$, for every $i \in \intset{n}$, one of the last $\len{n+1}$ bits of the binary representation of $i$ is $0$. So, $a^0$ has a distinct characterization in the power set formed by the gadgets $H^\Omega_i$. See Figure~\ref{fig:compressed_graph} for an example of $(H,B,r)$-compressed graph.

\begin{figure}
 \centering
 \includegraphics[width = .37\linewidth]{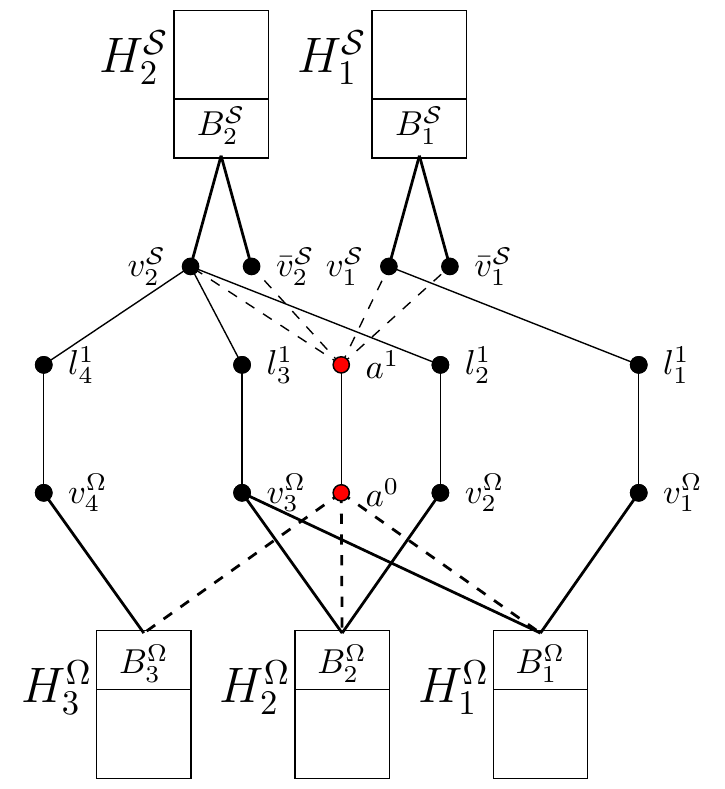}
 \caption{The $(H,B,2)$-compressed graph where $\Omega = \{1,2,3,4\}$ and \mbox{$\mathcal{S} = \{\{1,2\}, \{2,3,4\}\}$}.}%
 \label{fig:compressed_graph}
\end{figure}

\begin{proposition} \label{prop:graphsizeG}
 The graph $G = \Psi[H,B,r](\Omega,\mathcal{S})$ built on an instance $(\Omega,\mathcal{S})$ of \HS{}\begin{itemize}
  \item is connected and has size at most $|H|(\len{n+1}+\len{m})+r(n+1)+2m$,  where $|\Omega| = n$ and $|\mathcal{S}| = m$
  \item may be built in polynomial time in its size,
  \item is bipartite if the $B$-single extension of $H$ is bipartite.
 \end{itemize}
\end{proposition}

\begin{proof}
 The graph $G$ is formed by the union of $\len{n+1} + \len{m}$ copies of $H$, one vertex per variable, two vertices per clause, $n$ paths of $r-1$ vertices and one path of size $r$. Thus, in total, we have  \[(\len{n+1}+\len{m})|H|+n+2m+n(r-1)+r= |H|(\len{n+1}+\len{m})+r(n+1)+2m\]
 Finally, the two last items of the proposition are also direct by construction.
\end{proof}

\begin{theorem}
 \label{thm:setsizeG}
 Let $(\Omega,\mathcal{S})$ be an instance of \HS{} such that $|\Omega| = n$, $|\mathcal{S}| = m$. Let $(H,B,C)$ be a $(f,r)$-gadget for a $1$-\layered{} identifying function $f$ and let $(H',B',C')$ be a local $(g,q)$-gadget. The following propositions are equivalent: \begin{itemize}
  \item there exists a hitting set of $\mathcal{S}$ of size $k$.
  \item there exists a $(f,r)$-distance identifying set of $\Psi[H,B,1](\Omega,\mathcal{S})$ of size $k+|C|(\len{n+1}+\len{m})$.
  \item there exists a $(q,r)$-distance identifying set of $\Psi[H',B',q](\Omega,\mathcal{S})$ of size $k+|C'|(\len{n+1}+\len{m})$.
 \end{itemize}
\end{theorem}
\begin{proof}
 Suppose again that $P$ is a hitting set of $(\Omega,\Set)$ of size $k$. By denoting $C^\Omega_i$ and $C^\Set_j$ the copy of $C$ (respectively $C'$) associated to the copy $H^\Omega_i$ and $H^\Set_j$ of $H$ (respectively $H'$), we suggest the following set $I$ of size $k+|C|(\len{n+1}+\len{m})$ as a $(f,r)$-distance identifying set of $G = \Psi[H,B,1](\Omega,\Set)$ (respectively $(g,q)$-distance identifying set of $G' = \Psi[H',B',q](\Omega,\Set))$:
 $$ I = \{ v^\Omega_i \colon u_i \in P \} \; \cup \bigcup_{i \, \in \, \intset{\len{n+1}}} C^\Omega_i \; \cup \bigcup_{j \, \in \, \intset{\len{m+1}}}  C^\Set_j.$$
 
 By construction, $G$ is a $B^\Omega_i$-extension of $H^\Omega_i$ (respectively $B^\Set_j$-extension of $H^\Set_j$) for any $i \in \intset{\len{n+1}}$ (respectively $j \in \intset{\len{m}}$).
 This directly implies that $I$ is an $r$-dominating set of $G$ (respectively $q$-dominating set $G'$).
 
 We only have to show that $I$ $f$-distinguishes $G$ (respectively $G'$). Distinguishing the vertices of the gadget copies is still easy using the first item of Definition~\ref{def:gadgets}. Thereby, we only have to study the vertices of the form $v^\Set_j$, $\ov^\Set_j$, $l^k_i$, and $a^k$.
 To distinguish them, we mainly use the second item of Definition~\ref{def:gadgets} together with the characteristic function of the power set of the gadgets. We deduce that every couple of vertices is distinguished except when the two vertices are of the form $v^\Set_j$ or $\ov^\Set_{j}$ for $j \in \intset{m}$ (or if they are both of the form $a^k$ or $l^k_i$ for $k \in \intset{q-1}$ for $G'$, $G$ not containing such vertices).
 
 To distinguish $v^\Set_j$ and $\ov^\Set_{j}$ for $j \in \intset{m}$.
 We now use the fact that $P$ is a hitting set for $(\Omega,\Set)$.
 By definition of an hitting set, for any set $S_j \in \Set$, there exists a vertex $u_i \in P$ such that $u_i \in S_j$. We observe that $d(v^\Omega_i,v^\Set_j) = 1 < d(v^\Omega_i,\ov^\Set_j)$ by construction of $G$ (respectively $d(v^\Omega_i,v^\Set_j) = q < d(v^\Omega_i,\ov^\Set_j)$ by construction of $G'$) and that $v^\Omega_i \in I$ by definition of $I$. Since $f$ is $1$-\layered{} (respectively $g$ is $q$-local), $I$ $f$-distinguishes and $g$-distinguishes $v^\Set_j$ and $\ov^\Set_j$.
 
 For $G'$, it remains to distinguish $a^k$ and $l^k_i$ for $k \in \intset{q-1}$ and $i \in \intset{n}$. We recall that in a $(g,q)$-local gadget $(H',B',C')$, there exists $c \in C'$ such that $d(c,B') = k-1$ for each $k \in \intset{q}$. Then we may use characteristic function of the power set together with property $(\beta_1)$ of a $q$-local function to distinguish them.\medskip
 
 In the other direction, assume that $I$ is a $(f,r)$-distance identifying set of $G$ of size $k+|C|(\len{n+1}+\len{m})$ (respectively $(g,q)$-distance identifying set of $G$ of size $k+|C'|(\len{n+1}+\len{m})$).
 As every set of $S$ is not empty, we may define a function $\varphi : \intset{m} \to \intset{n}$ such that $u_{\varphi(j)} \in S_j$.
 
 We suggest the following set $P$ as an hitting set of $\mathcal{S}$ of size at most $k$:
 $$
  P = \{ u_i \in \Omega \; | \; l^k_i \in I \text{ for any } k \in \intset{r-1} \} \cup \{ u_{\varphi(j)} \in \Omega \; | \; v^\Set_j \in I \text{ or } \ov^\Set_j \in I\}
 $$
 
 The size of $P$ is ensured by the fourth item of Definition~\ref{def:gadgets} of a gadget.
 
 We claim that the only vertices that may $f$-distinguish the couple $(v^\Set_j, \ov^\Set_j)$ are themselves and the vertices of form $l^k_i$ such that $u_i \in S_j$. To prove so, we apply propriety $(\alpha)$ from Definition~\ref{def:dif} (a vertex cannot distinguish two vertices at the same distance from it) on the following enumeration on $G$: \begin{itemize}
  \item the vertices $a^k$ verify $d(a^k,v^\Set_j) = r-k = d(a^k,\ov^\Set_j)$ for each $k \in \intset{r-1}$
  \item a vertex $l^k_i$ such that $u_i \not \in S_j$ verifies $d(l^k_i,v^\Set_j) = 2+r-k  = d(l^k_i,\ov^\Set_j)$ because of $a^{r-1}$
  \item a vertex $v$ of $H^\Omega_i$ verifies $d(v,v^\Set_j) = d(v,B^\Omega_i)+1+r = d(v,\ov^\Set_j)$ because of the path formed by the vertices of form $a^k$.
  \item a vertex $v$ of $H^\Set_{j'}$ with $j \neq j'$ verifies $d(v,v^\Set_j) = d(v,B^\Set_{j'})+3 = d(v,\ov^\Set_j)$ because of $a^{r-1}$
  \item both $v^\Set_j$ and $\ov^\Set_j$ are $B^\Set_j$-adjacent, so they are at the same distance of any vertex of $H^\Set_j$.
 \end{itemize}
 The enumeration on $G'$ is identical when replacing $r$ by $q$.
 We deduce that $v^\Set_j$ and $\ov^\Set_j$ are $f$-distinguished if and only if either one on them belongs to $I$ (in that case $u_{\varphi(j)} \in P \cap S_j$) or if there exists $l^k_i \in I$ such that $u_i \in S_j$ and $k+1 \in \intset{r}$ (and then $u_i \in P \cap S_j$).
\end{proof}

\section{On Providing Gadgets to Establish Generic Reductions}\label{sec:cis}

In this section, we finalize the reductions by furnishing some gadgets and combining them with the suitable theorems and propositions from Section~\ref{sec:sup}. The existence of the gadgets rely on the following tool lemma:
\begin{lemma}[Twins Lemma]
 \label{lemma:twin}
 Let $x$ and $y$ be two vertices of a graph $G$ such that \mbox{$N(x) = N(y)$}.
 Then any distance identifying set of $G$ contains either $x$ or $y$.
\end{lemma}
\begin{proof}
 Because $N(x) = N(y)$, for every vertex $u$ of $G$, if $u \not \in \{x,y\}$, then $d(u,x) = d(u,y)$.
 Thus, by property $(\alpha)$ of a distance identifying set, $u$ may distinguishes $x$ and $y$ if and only if $u \in \{x,y\}$, implying that a distance identifying set must contain either $x$ or $y$.
\end{proof}

The gadgets are defined as follows:

\begin{definition}[The $1$-\layered{} gadget] Let $H$ be the bipartite planar graph such that:\begin{itemize}
  \item Its ten vertices are denoted $b$, $\bar{b}$, $u_1$, $\bar{u}_1$, $u_2$, $\bar{u}_2$, $v_1$, $\bar{v}_1$, $v_2$ and $\bar{v}_2$,
  \item The vertices $u_1, u_2, \bar{u}_1$ and $\bar{u}_2$ form a cycle as well as the vertices $v_1, v_2, \bar{v}_1$ and $\bar{v}_2$.
  \item The vertices $b$ and $\bar{b}$ are adjacent to $u_1$, $\bar{u}_1$, $v_1$ and $\bar{v}_1$.
 \end{itemize}
 We define the sets $B = \{b,\bar{b}\}$ and $C = \{b, u_1, u_2, v_1, v_2\}$.\\
 The triple $(H,B,C)$ is called \textit{the $1$-\layered{} gadget} (see Fig.~\ref{fig:1cgadget}).
\end{definition}

\begin{figure}
 \centering
 \includegraphics[width = 0.32\linewidth]{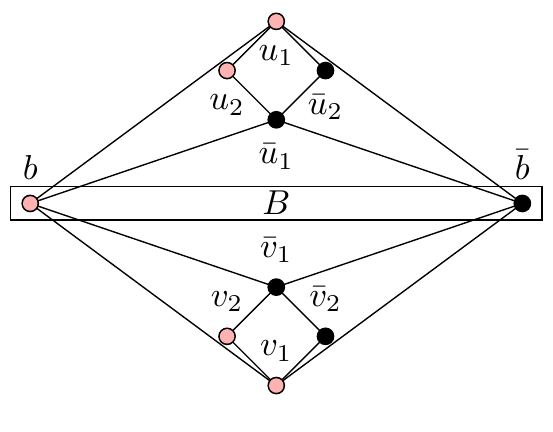}
 \caption{The $1$-\layered{} gadget $(H,B,C)$.  $C$ contains the colored vertices.}\label{fig:1cgadget}
\end{figure}

\begin{definition}[The $r$-local $0$-\layered{} gadget] Given an integer $r > 1$ (respectively $r = 1$), let $H_r$ be the bipartite planar graph of size $4r+2$ (respectively 8) such that: \begin{itemize}
  \item its vertices are denoted $a_i$ and $b_i$ for $i \in \intset{2r+1}$ (respectively $i \in \intset{4}$),
  \item for each $i \in \intset{2r}$ (respectively $i \in \intset{3}$), both $a_i$ and $b_i$ are adjacent to $a_{i+1}$ and $b_{i+1}$
 \end{itemize}
 We define the sets $B_r = \{a_1,b_1\}$ and $C_r = \{a_i \ | \ i \in \intset{2r+1}\}$ (respectively $C_1 = \{a_1,a_2,a_3,a_4\}$). The triple $(H_r,B_r,C_r)$ is called \textit{the $r$-local $0$-\layered{} gadget} (see Fig.~\ref{fig:rl0cgadget}).
\end{definition}

For each positive integer $r$, $r$-\aLD{} and $r$-\aMD{} are $r$-local $0$-\layered{} problems, whereas $r$-\aIC{} is not $0$-\layered{}. We define a specific gadget for this remaining problem.

\begin{definition}[The $r$-\aIC{} gadget] Given a positive integer $r$, let $H_r$ be the bipartite planar graph of size $6r+4$ such that: \begin{itemize}
  \item its vertices are denoted $a_{i-1}$ and $b_{i-1}$ for $i \in \intset{r+2}$, and $a^j_i$ and $b^j_i$ for $i \in \intset{r}$ and $j \in \intset{2}$. We also denote $a_0$ as $a^1_{r+1}$ and $a^2_{r+1}$ and we denote $b_0$ as $b^1_{r+1}$, $b^2_{r+1}$, $a^1_0$ and $a^2_0$.
  \item the edges are all included in the six following paths \begin{itemize}
         \item from $a_0$ to $a_{r+1}$ such that $d(a_0,a_i) = i$ for $i \in \intset{r+1}$.
         \item from $b_0$ to $b_{r+1}$ such that $d(b_0,b_i) = i$ for $i \in \intset{r+1}$.
         \item from $a_0^1$ to $a_{r+1}^1$ such that $d(a_0^1,a_i^1) = i$ for $i \in \intset{r+1}$.
         \item from $a_0^2$ to $a_{r+1}^2$ such that $d(a_0^2,a_i^2) = i$ for $i \in \intset{r+1}$.
         \item from $b_1^1$ to $b_{r+1}^1$ such that $d(b_1^1,b_i^1) = i-1$ for $i \in \intset{r+1}$.
         \item from $b_1^2$ to $b_{r+1}^2$ such that $d(b_1^2,b_i^2) = i-1$ for $i \in \intset{r+1}$.
        \end{itemize}
 \end{itemize}
 We define the sets $B_r = \{b_1^1,b_1^2\}$ and $C_r = \{a_{r+1}, b_{r+1}\} \cup \underset{i \in \intset{r+1}}{\bigcup} \{a_i^1,b_i^1\}$.\\
 The triple $(H_r,B_r,C_r)$ is called \textit{the $r$-\aIC{} gadget} (see Fig.~\ref{fig:IC_gadget}).
\end{definition}

\begin{figure}
 \begin{minipage}{0.47\linewidth}
  \centering
  \includegraphics[width = 0.6\linewidth]{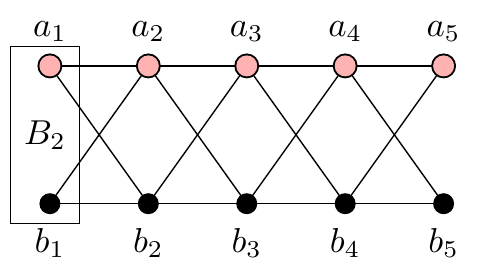}%
  \caption{The $2$-local $0$-\layered{} gadget. $C_2$ contains the colored vertices.}%
  \label{fig:rl0cgadget}%
 \end{minipage} \hfill%
 \begin{minipage}{0.47\linewidth}
  \centering
  \includegraphics[width = 0.7\linewidth]{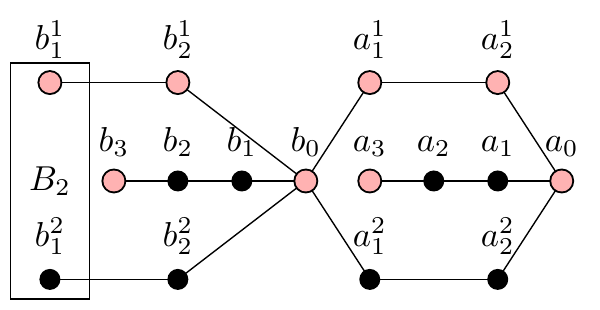}%
  \caption{The $2$-\aIC{} gadget. $C_2$ contains the colored vertices.}%
  \label{fig:IC_gadget}%
 \end{minipage}%
\end{figure}

As expected, we have the following propositions:
\begin{proposition} \label{prop:1cgadget} The $1$-\layered{} gadget is a bipartite planar $(f,r)$-gadget for any $1$-\layered{} distance identifying function $f$ and $r \in \inftyset$.
\end{proposition}
\begin{proof}
 We have to check the four conditions to be a $(f,r)$-gadget.
 Consider a $B$-extension $G$ of $H$.
 Clearly, $(p_d)$ is satisfied as $C$ is even a $1$-dominating set of $V_H \cup N_B$. The condition $(p_s)$ is also easily verified using the Twins Lemma~\ref{lemma:twin} on the distinct pairs $(b,\bar{b})$, $(u_1,\bar{u}_1)$, $(u_2,\bar{u}_2)$, $(v_1,\bar{v}_1)$ and $(v_2,\bar{v}_2)$. To prove $(p_h)$ and $(p_b)$, we only have to study vertices not belonging to $C$ (as $f$ is $0$-\layered{}).
 Remark that:\begin{itemize}
  \item $\bar{u}_1$ is the only vertex outside of $C$ that is adjacent to $u_2$.
  \item $\bar{v}_1$ is the only vertex outside of $C$ that is adjacent to $v_2$
  \item $\bar{u}_2$ is the only vertex outside of $C$ that is adjacent to $u_1$ and not adjacent to $v_1$.
  \item $\bar{v}_2$ is the only vertex outside of $C$ that is adjacent to $v_1$ and not adjacent to $u_1$.
  \item $\bar{b}$ is the only vertex outside $C$ both adjacent to $u_1$ and $v_1$
  \item a $B$-adjacent vertex is not adjacent to $u_1$ nor to $v_1$ but is adjacent to $b$
  \item finally, a vertex from $V_{G \setminus H}$ which is not $B$-adjacent is neither adjacent to $u_1$, $v_1$ nor $b$.
 \end{itemize}
 Therefore properties $(p_h)$ and $(p_b)$ are satisfied by $(H,B,C)$ which is a $1$-\layered{} gadget for $f$.
\end{proof}

\begin{proposition} \label{prop:rl0cgadget} Given a positive integer $r$, the $r$-local $0$-\layered{} gadget is a local bipartite planar $(f,r)$-gadget for every $r$-local $0$-\layered{} distance identifying function $f$.
\end{proposition}
\begin{proof}
 We have to check the five conditions to be a local $(f,r)$-gadget. Clearly, $(p_d)$ is satisfied as $C_r$ is even a $1$-dominating set of the $B_r$-single-extension of $H_r$. The condition $(p_s)$ is also easily verified using the Twins Lemma~\ref{lemma:twin} on the distinct couples $(a_i,b_i)$ for $i \in \intset{r+1}$ (or $i \in \intset{4}$ if $r = 1$). The local condition $(p_l)$ is satisfied as $d(a_i,B_r) = i-1$ for every $i \in \intset{2r+1}$. To prove $(p_h)$ and $(p_b)$, we only have to study vertices of form $b_i$ (as $f$ is $0$-\layered{}). When $r > 1$: \begin{itemize}
  \setlength\itemsep{0.1em}
  \item the $r$-neighbourhood of $b_{r+1}$ is $C$.
  \item for every $i \in \intset{r}$ the $r$-neighbourhood of $b_{i}$ is $\{a_j \; | \; j \in \intset{i+r}\}$, and the $r$-neighbourhood of $b_{2r+2-i}$ is $\{a_{2r+2-j} \; | \; j \in \intset{i+r}\}$
  \item the $r$-neighbourhood of a $B_r$-adjacent vertex is $\{a_i \; | \; i \in \intset{r}\}$
  \item finally, the $r$-neighbourhood of a vertex outside of $H_r$ which is not $B_r$-adjacent is the set \mbox{$\{a_i \; | \; i \in \intset{r-1}\}$}
 \end{itemize}
 In the specific case where $r = 1$: \begin{itemize}
  \setlength\itemsep{0.1em}
  \item $b_4$ is adjacent to $a_3$
  \item $b_3$ is adjacent to $a_2$ and $a_4$
  \item $b_2$ is adjacent to $a_1$ and $a_3$
  \item $b_1$ is adjacent to $a_2$
  \item a $B_1$-adjacent vertex is adjacent to $a_1$
  \item a vertex outside of $H_1$ which is not $B_1$-adjacent is not adjacent to any vertex from $C$.
 \end{itemize}
\end{proof}

\begin{proposition} \label{prop:rlgadget} Given a positive integer $r$, the $r$-\aIC{} gadget is a local bipartite planar $(f,r)$-gadget for the identifying function $f$ associated with $r$-\aIC{}, where $f_G[w](u,v) = \true$ if $w\in N_r[u] \Delta N_r[v]$ for relevant inputs $G$, $u$, $v$ and $w$.
\end{proposition}

\begin{proof}
 Consider $G$ a $\BIC$-extension of $\HIC$.
 The set of $\BIC$-adjacent vertices of $G$ is denoted $N_{\BIC}$.
 Clearly property  $(p_d)$ is satisfied as $C$ is an $r$-dominating set of \mbox{$G[V_{\HIC} \cup N_{\BIC}]$}.
 
 First, let us prove that $(\HIC,\BIC,\CIC)$ verifies property $(p_b)$, that is $\CIC$ distinguishes $V_{\HIC}$ and $V_G$.
 
 Observe that the only vertices $x$ such that $d(x,b_{r+1}) \le r$ are the vertices $b_i$ for $i \in \intset{r+1}$. So the pairs $(b_i, x)$ for $x \not= b_j$ for some $j \in \intset{r+1}$ are distinguished by $b_{r+1}$.
 In a similar way, we show that the pairs $(a_i,x)$ for $x \not= a_j$ for a $j \in \intset{r+1}$ are distinguished by $a_{r+1}$.
 
 Let $i$ and $j$ be two integers such that $0\leq i < j\leq r$, and suppose that we have two vertices $x$ and $y$ such that $x$ is either $b^1_j$ or $b^2_j$, and $y$ is either $b^1_i$ or $b^2_i$. Then $a^1_{j-1}$ distinguishes $x$ and $y$ because $a^1_{j-1} \in N_r[y] \Delta N_r[x]$. The same reasoning holds when $x \in \{a^1_j, a^2_j\}$ and $y \in \{a^1_i ,a^2_i\}$. Furthermore, we have $N_r[b_{r+1}] \cap \CIC = \{b_{r+1}\}$ and there are no other vertices $x$ such that $N_r[x] \cap \CIC = \{b_{r+1}\}$. So $b_{r+1}$ is distinguished from any other vertices of $G$. The same reasoning proves that $a_{r+1}$ is distinguished from any other vertices too.
 
 The vertex $a_0$ distinguishes any vertex of $\{b_0\}\cup \{b^1_i, b^2_i \colon i \in \intset{r}\}$ from any vertex of $\{a_0\}\cup \{a^1_i, a^2_i \colon i \in \intset{r}\}$.
 
 Now, given any vertex $x$ in $\{a^1_i, a^2_i \colon i\in \intset{r}\}$ (resp. $(b^1_i, b^2_i \colon i\in \intset{r}$), one can notice that $b^1_r$ (resp. $a^1_r$) distinguishes $x$ and $b_0$.
 
 Given an integer $i\in\intset{r}$, $b^1_{r-i+1}$ distinguishes $b^1_i$ and $b^2_i$ because $d(b^1_i,b^1_{r-i+1})=|r-2i+1| \in \intset{r-1}$ and $d(b^1_{r-i+1}, b^2_i) > r$. In the same way, $a^1_i$ and $a^2_i$ are distinguished by $a^1_{r-i+1}$. The vertices $a_0$ and $a_i$ for $i\in \intset{r}$ are distinguished by $a^1_1$. So, at this stage, we proved that $\CIC$ is a distinguishes $\HIC$.

 Let $x$ be a vertex of $V_{G\setminus \HIC}$ and $y$ a vertex of $V_{\HIC}$. If $y\in \{b^1_i, b^2_i, b_i \colon i\in \intset{r}\} \cup \{b_0\}$, then $b_0 \in N_r[x] \Delta N_r[y]$ and so $x$ and $y$ are distinguished. Otherwise, if $y$ is $b_{r+1}$ (resp. $a_{r+1}$), then $b_{r+1}$ (resp. $a_{r+1}$) distinguishes $x$ and $y$. Otherwise, $a_0$ distinguishes the two vertices. Then $(\HIC,\BIC, \CIC)$ verifies property $(p_h)$.
 
 Let us now prove that $\HIC$ distinguishes $N_{\BIC}$ and $V_{G\setminus\HIC} \setminus N_{\BIC}$. Let $x$ be a vertex of the former and $y$ be a vertex of the latter. Then $b_0 \in N_r[x] \Delta N_r[y]$ since $N_{\BIC}$ contains the only vertices with neighbours in $\HIC$. The vertices $x$ and $y$ are distinguished.

 Now, let $S$ be a $(f,r)$-distance identifying set of $G$, where $f$ is defined such that $f_G[w](u,v) = \true$ if $w\in N_r[u] \Delta N_r[v]$ for relevant inputs $G$, $u$, $v$ and $w$. We want to prove property $(ps)$ \emph{i.e.} we have $|\CIC| \leq |S \cap V_{\HIC}|$.  As $N_r[b_r] \Delta N_r[b_{r+1}] = \{ b_0 \}$, then $b_0$ is the only vertex which can distinguish $b_r$ and $b_{r+1}$, then $b_0 \in S$. Similarly with $a_r$ and $a_{r+1}$, we must have $a_0 \in S$.
 For every $i \in \intset{r}$, we have $N_r[b^1_i] \Delta N_r[b^2_i] = \{ b^1_{r-i},b^2_{r-i} \}$ then either $b^1_{r-i}$ or $b^2_{r-i}$ must be in $S$. The same reasoning on $a^1_i$ and $a^2_i$ implies that either $a^1_{r-i}$ or $a^2_{r-i}$ must be in $S$.
 Furthermore, as $N_r[b_{r+1}] \cap S$ (resp. $N_r[a_{r+1}] \cap S $) cannot be empty by definition of an identifying code, then there exists $i\in[r+1]$ such that $b_i \in S$ (resp. $a_i \in S$).
 We conclude that there are at least $2r+4$ vertices of $\HIC$ (that is the size of $\CIC$) in $S$, proving that property $(p_s)$ holds. This proves that $(\HIC,\BIC, \CIC)$ is a $(f, r)$-gadget. By construction, we can easily see that $(\HIC,\BIC, \CIC)$ verifies the property of local gadgets. Therefore, this is also a local-gadget.

\end{proof}
With Propositions~\ref{prop:1cgadget} to~\ref{prop:rlgadget}, we can now prove the Theorems~\ref{thm:CPlowerbound} to~\ref{thm:Glowerbound}.
\begin{proof}[Proof of Theorems~\ref{thm:CPlowerbound} and~\ref{thm:Glowerbound} for each $1$-\layered{} identifying function $f$ and $r \in \inftyset$] ~\\ We first suggest a reduction from \textsc{Planar} \HS{} to $(f,r)$-\aDIS{} based on the bipartite planar $1$-\layered{} gadget $(H,B,C)$. Let $(\Omega,\mathcal{S})$ be an instance of \textsc{Planar} \HS{} with $|\Omega| = n$ and $|\mathcal{S}| = m$ such that $m = \mathcal{O}(n)$. According to Proposition~\ref{prop:graphsize}, the bipartite apex graph $G = \Phi^*[H,B](\Omega,\mathcal{S})$ has size $n'$ linear in $n+m = \mathcal{O}(n)$ and may be built in polynomial-time in its size.
 Recall that $(H,B,C)$ is a $(f,r)$-gadget by Proposition~\ref{prop:1cgadget}. By Theorem~\ref{thm:setsize}, $G$ admits a $(f,r)$-distance identifying set of size $k' = k+|C|(n+m)$ if an only if $\mathcal{S}$ admits an hitting set of size $k$.
 Thus, an algorithm solving $(f,r)$-\aDIS{} in $2^{o(\sqrt{n'})}$ would solve \textsc{Planar} \HS{} in time $2^{o(\sqrt{n})}$, a contradiction to Theorem~\ref{thm:planarHS} (assuming ETH).
 
 We adapt the previous argumentation to get a reduction from \HS{} to $(f,r)$-\aDIS{}, the instance $(\Omega,\mathcal{S})$ belonging now to the \HS{} problem. According to Proposition~\ref{prop:graphsizeG}, the bipartite graph $G = \Psi[H,B,1](\Omega,\mathcal{S})$ has size $n'$ linear in $n+m = \mathcal{O}(n)$ and may also be built in polynomial-time in its size.
 By Theorem~\ref{thm:setsizeG}, $G$ admits a $(f,r)$-distance identifying set of size $k' = k+|C|(\len{n+1}+l_m)$ if an only if $\mathcal{S}$ admits an hitting set of size $k$.
 Thus, an algorithm solving $(f,r)$-\aDIS{} in $2^{o(n')}$ would solve \HS{} in time $2^{o(n)}$, contradicting Theorem~\ref{thm:hs} when assuming ETH. Moreover, a parameterized algorithm solving $(f,r)$-\aDIS{} in $2^{\mathcal{O}(k)} \cdot {n'}^{\mathcal{O}(1)}$ would be in contradiction with Theorem~\ref{thm:w2bound} when assuming $\W{2} \neq \FPT{}$.
\end{proof}

\begin{proof}[Proof of Theorems~\ref{thm:LPlowerbound} and~\ref{thm:Glowerbound} for each $r$-local identifying function $f$] First, we suggest a reduction from \textsc{Planar} \HS{} to $(f,r)$-\aDIS{}. Assuming the existence of a (bipartite) planar local $(f,r)$-gadget $(H,B,C)$. Let $(\Omega,\mathcal{S})$ be an instance of \textsc{Planar} \HS{} with $|\Omega| = n$ and $|\mathcal{S}| = m$ such that $m = \mathcal{O}(n)$. According to Proposition~\ref{prop:graphsize}, the (bipartite) planar graph $G = \Phi[H,B,r](\Omega,\mathcal{S})$ has size $n'$ linear in $n+m = \mathcal{O}(n)$ and may be built in polynomial-time in its size.
 By Theorem~\ref{thm:setsize}, $G$ admits a $(f,r)$-distance identifying set of size $k' = k+|C|(n+m)$ if and only if $\mathcal{S}$ admits an hitting set of size $k$.
 Thus, an algorithm solving $(f,r)$-\aDIS{} in $2^{o(\sqrt{n'})}$ would solve \textsc{Planar} \HS{} in time $2^{o(\sqrt{n})}$, a contradiction to Theorem~\ref{thm:planarHS} (assuming ETH).
 
 We adapt the previous argumentation to get a reduction from \HS{} to $(f,r)$-\aDIS{}, the instance $(\Omega,\mathcal{S})$ belonging now to the \HS{} problem. In this case, we only have to assume the existence of a (bipartite) local $(f,r)$-gadget $(H,B,C)$. According to Proposition~\ref{prop:graphsizeG}, the bipartite graph $G = \Psi[H,B,r](\Omega,\mathcal{S})$ has size $n'$ linear in $n+m = \mathcal{O}(n)$ and may also be built in polynomial-time in its size.
 By Theorem~\ref{thm:setsizeG}, $G$ admits a $(f,r)$-distance identifying set of size $k' = k+|C|(\len{n+1}+l_m)$ if and only if $\mathcal{S}$ admits an hitting set of size $k$.
 Thus, an algorithm solving $(f,r)$-\aDIS{} in $2^{o(n')}$ would solve \HS{} in time $2^{o(n)}$, contradicting Theorem~\ref{thm:hs} when assuming ETH. Moreover, a parameterized algorithm solving $(f,r)$-\aDIS{} in $2^{\mathcal{O}(k)} \cdot {n'}^{\mathcal{O}(1)}$ would be in contradiction with Theorem~\ref{thm:w2bound} when assuming $\W{2} \neq \FPT{}$.
\end{proof}

\begin{proof}[Proof of Theorems~\ref{thm:LCPlowerbound} and~\ref{thm:Glowerbound} for each $r$-local $0$-\layered{} identifying function $f$] \hspace{0.3em}
 By Proposition~\ref{prop:rl0cgadget}, the $r$-local $0$-\layered{} gadget is a local bipartite planar $(f,r)$-gadget for every $r$-local $0$-local distance identifying function $f$. Then, Theorems~\ref{thm:LPlowerbound} and~\ref{thm:Glowerbound} for $r$-local identifying function $f$ apply and directly yield the current theorems.
\end{proof}

\section{Conclusion}

In this paper, we showed generic tools to analysis identifying problems and their computational lower bounds.
This study opens some new questions. First of all, we observe that our toolbox does not contain a $r$-local gadget. Does one exist? Furthermore, there is still a gap between the computational lower bound provided by Theorem~\ref{thm:Glowerbound} and the elementary upper bound from Proposition~\ref{prop:upperbound} in the local cases. We wonder if local problems may be solved in $k^{\mathcal{O}(k)} \cdot n^{\mathcal{O}(1)}$. Notice that a polynomial kernel would imply such a complexity (but the reciprocal is not true).
For non-local problems, an \FPT{} upper bound is globally unknown. In particular, $\W{2}$-hard problems like \aMD{} cannot admit \FPT{} algorithms unless $\W{2} = \FPT{}$. Then, which non-local problem is $\W{2}$-hard?
We mention that we actually get a \FPT{} reduction from \HS{} to some scarce non-local problems (however including \aMD{}) proving their $\W{2}$-hardness, but the family of involved problems is not precise nor wide.
Nevertheless, we remark that most of our reductions may be generalized to the oriented version of \DIS{} sometimes even for the strongly connected graphs \---this is due to the fact that the paths in our distance identifying graphs and gadgets may often be seen as oriented\---. Thus, we inform the community that the oriented version of \aMD{} (studied for Cayley graphs in \cite{FehrGO06}) remains $\W{2}$-hard.

\end{document}